\documentclass{article}

\usepackage[preprint]{neurips_2024}

\usepackage[utf8]{inputenc} % allow utf-8 input
\usepackage[T1]{fontenc}    % use 8-bit T1 fonts
\usepackage{hyperref}       % hyperlinks
\usepackage{url}            % simple URL typesetting
\usepackage{booktabs}       % professional-quality tables
\usepackage{amsfonts}       % blackboard math symbols
\usepackage{nicefrac}       % compact symbols for 1/2, etc.
\usepackage{microtype}      % microtypography
\usepackage{graphicx}
\usepackage{wrapfig}
\usepackage{amsmath}
\usepackage{multirow}
\usepackage[table,xcdraw]{xcolor}
\usepackage{colortbl}      % Additional color options for tables

\usepackage{siunitx}       % For number formatting

\definecolor{lightgray}{gray}{0.9}
\definecolor{red}{rgb}{1,0,0}

\newcommand{\gain}[1]{\textcolor{red}{(+#1)}}

\usepackage{microtype}
\usepackage{subfigure}
\usepackage[linesnumbered, ruled, vlined]{algorithm2e}
\usepackage{caption}
\usepackage{authblk}
\usepackage{hyperref}
\usepackage{amsmath}
\usepackage{amssymb}
\usepackage{mathtools}
\usepackage{amsthm}

\usepackage[capitalize,noabbrev]{cleveref}
\theoremstyle{plain}
\newtheorem{theorem}{Theorem}[section]

\theoremstyle{definition}

\newtheorem{assumption}[theorem]{Assumption}
\theoremstyle{remark}

\usepackage[textsize=tiny]{todonotes}

\title{UniCompress: Enhancing Multi-Data Medical Image Compression with Knowledge Distillation}

\author[1]{Runzhao Yang}
\author[2,\thanks{Equal Contribution.}]{Yinda Chen}
\author[1]{Zhihong Zhang}
\author[2]{Xiaoyu Liu}
\author[3]{Zongren Li}
\author[3]{Kunlun He}
\author[2]{Zhiwei Xiong}
\author[1]{Jinli Suo}
\author[1]{Qionghai Dai}

\affil[1]{Tsinghua University}
\affil[2]{University of Science and Technology of China}
\affil[3]{301 Hospital}

\begin{document}

\maketitle

\begin{abstract}
In the field of medical image compression, Implicit Neural Representation (INR) networks have shown remarkable versatility due to their flexible compression ratios, yet they are constrained by a one-to-one fitting approach that results in lengthy encoding times. Our novel method, ``\textbf{UniCompress}'', innovatively extends the compression capabilities of INR by being the first to compress multiple medical data blocks using a single INR network. By employing wavelet transforms and quantization, we introduce a codebook containing frequency domain information as a prior input to the INR network. This enhances the representational power of INR and provides distinctive conditioning for different image blocks. Furthermore, our research introduces a new technique for the knowledge distillation of implicit representations, simplifying complex model knowledge into more manageable formats to improve compression ratios. Extensive testing on CT and electron microscopy (EM) datasets has demonstrated that UniCompress outperforms traditional INR methods and commercial compression solutions like HEVC, especially in complex and high compression scenarios. Notably, compared to existing INR techniques, UniCompress achieves a 4$\sim$5 times increase in compression speed, marking a significant advancement in the field of medical image compression.
Codes will be publicly available.
% Code is available at \url{https://anonymous.4open.science/r/UniCompress-933C/}.
\end{abstract}

\section{Introduction}
Implicit Neural Representations (INRs), which leverage neural networks to approximate continuous functions within complex data structures \cite{sitzmann2020implicit, sitzmann2019scene,zhang2021generator}, were initially applied in scene rendering and reconstruction \cite{pumarola2021d, martin2021nerf}. Recent advancements have extended the use of INRs to lossy compression, achieving general, modality-agnostic compression by overfitting small networks \cite{yang2023sci, yang2023tinc}. INRs have demonstrated comparable performance to traditional methods in terms of objective metrics like PSNR and SSIM as well as image quality \cite{guo2023compression,tang2024z,tang2024zerothorder}. However, in large-scale medical datasets, INR methods lag behind traditional methods \cite{yang2022sharing}, primarily due to their local fitting approach, which often overlooks high-frequency information in images, resulting in excessive smoothing. Additionally, the one-to-one fitting approach of INRs leads to longer encoding times, posing significant challenges in large-scale medical image compression applications.

In this paper, we first theoretically analyze and compare the bounds on reconstruction errors for INR and VAE. We derive that if the INR decoder generates images closer to the original images with smaller magnitudes, INR will have a lower reconstruction error. Based on these insights, we introduce UniCompress, the first method capable of simultaneously compressing multiple data blocks, aiming to reduce encoding time and leverage existing knowledge. Our approach begins with a wavelet transform of medical image blocks, decomposing them into high-frequency and low-frequency components \cite{villasenor1995wavelet, liu2018multi,zhang2024deepgi}. Robust visual features are then extracted using networks pretrained on large-scale medical images \cite{chen2023generative,mo2024password,pmlr-v216-tang23a}. We employ an attention fusion module to integrate features from different modalities and use a learnable codebook \cite{esser2021taming} to quantize the extracted multimodal features. This process combines discrete feature values $\mathcal{Z}$ with the numerical inputs of the INR, facilitating the parallel learning of multiple image block features and incorporating frequency domain information as priors.

To further enhance our method, we implement knowledge distillation \cite{wang2019private}, where a smaller student network learns prior representations extracted by a larger teacher network. By aligning these representations through consistency in discrete codebook features and intermediate features in the reconstructed image space, we can significantly improve the compression ratio while maintaining visual quality.

Our comprehensive evaluation shows that using a simple MLP SIREN network, our method achieves an average PSNR improvement of \textbf{0.95 dB} and \textbf{1.71 dB} on CT and EM datasets, respectively, at the same compression ratio. Further distillation results yield performance gains of \textbf{0.71 dB} and \textbf{0.77 dB}. Moreover, we find that even cross-modal knowledge distillation aids performance, with out-of-domain knowledge providing more significant benefits in simpler datasets like Liver. Our method increases data encoding speed by 4$\sim$5 times, emphasizing efficient parallel processing across multiple datasets.

In summary, our contributions in this paper are as follows:

(1) We establish a theoretical analysis comparing the error bounds of INR and VAE, deriving sufficient conditions under which INR achieves lower reconstruction errors if the INR decoder generates images closer to the original with smaller magnitudes.

(2)We propose the first method that employs Implicit Neural Representations (INRs) to simultaneously represent multiple datasets. By incorporating a discrete codebook prior with frequency domain information, our network can represent diverse medical datasets concurrently.

(3) We introduce a two-stage training method using knowledge distillation, which transfers learned representations from complex models to simpler, smaller models, significantly enhancing inference speed and compression ratio.

(4) We conduct experiments on complex Electron Microscopy and CT data. Despite using only a basic SIREN scheme, our UniCompress demonstrates robust performance, validating its effectiveness in handling challenging datasets.

\section{Related Work}
\textbf{Implicit Neural Representation.}
Recent advances in Implicit Neural Representations (INR) have shown significant progress. SIREN \cite{sitzmann2020implicit} uses periodic activation functions for enhanced scene rendering. WIRE \cite{saragadam2023wire} employs wavelet-based INR with complex Gabor wavelets for better visual signal representation and noise reduction. INCODE \cite{kazerouni2024incode} integrates image priors for adaptive learning of activation function parameters.

In medical imaging, INRs like SIREN, ReLU, Snake, Sine+, and Chirp are effective for brain image registration by estimating displacement vectors and deformation fields \cite{byra2023exploring}. For video frame interpolation (VFI), INRs align signal derivatives with optical flow to maintain brightness, outperforming some state-of-the-art models \cite{zhuang2022optical}.

\textbf{Medical Imaging Prior.}
Effective feature extraction enhances network performance. Zhou et al. \cite{zhou2023xnet} used wavelet-transformed spectral features for segmentation, focusing on high-frequency areas. Chen et al. \cite{chen2023self} used Histogram of Oriented Gradients (HOG) to balance performance and computational complexity. Self-supervised methods like multi-scale contrastive learning \cite{chen2023learning} and text-image feature alignment \cite{chen2023generative,chen2024bimcv} have also shown progress.

\textbf{Medical Image Compression.}
Efficient compression is crucial for large medical datasets. Traditional methods like JPEG 2000 and JP3D use wavelet transforms for optimal performance and complexity \cite{bruylants2015wavelet}. Xue et al. \cite{xue2023dbvc} achieve efficient 3D biomedical video encoding through motion estimation and compensation.
Variational Autoencoders (VAEs) are also explored for lossy compression and structured MRI compression \cite{van2017neural}. QARV \cite{duan2023qarv} uses quantization-aware ResNet VAE, while Attri-VAE \cite{cetin2023attri} provides attribute-based interpretable representations.
INR-based methods like SCI \cite{yang2023sci} and TINC \cite{yang2023tinc} offer flexible block-wise approaches. COMBINER \cite{guo2023compression} compresses data as INRs using relative entropy encoding. Despite their flexibility, INRs face challenges like low processing efficiency and high-frequency information loss. Our approach integrates medical image characteristics into INR inputs and processes multiple data blocks simultaneously, an unexplored area in INR compression for large-scale medical imaging.
\begin{figure*}[t]
    \centering
    \includegraphics[width =\linewidth]{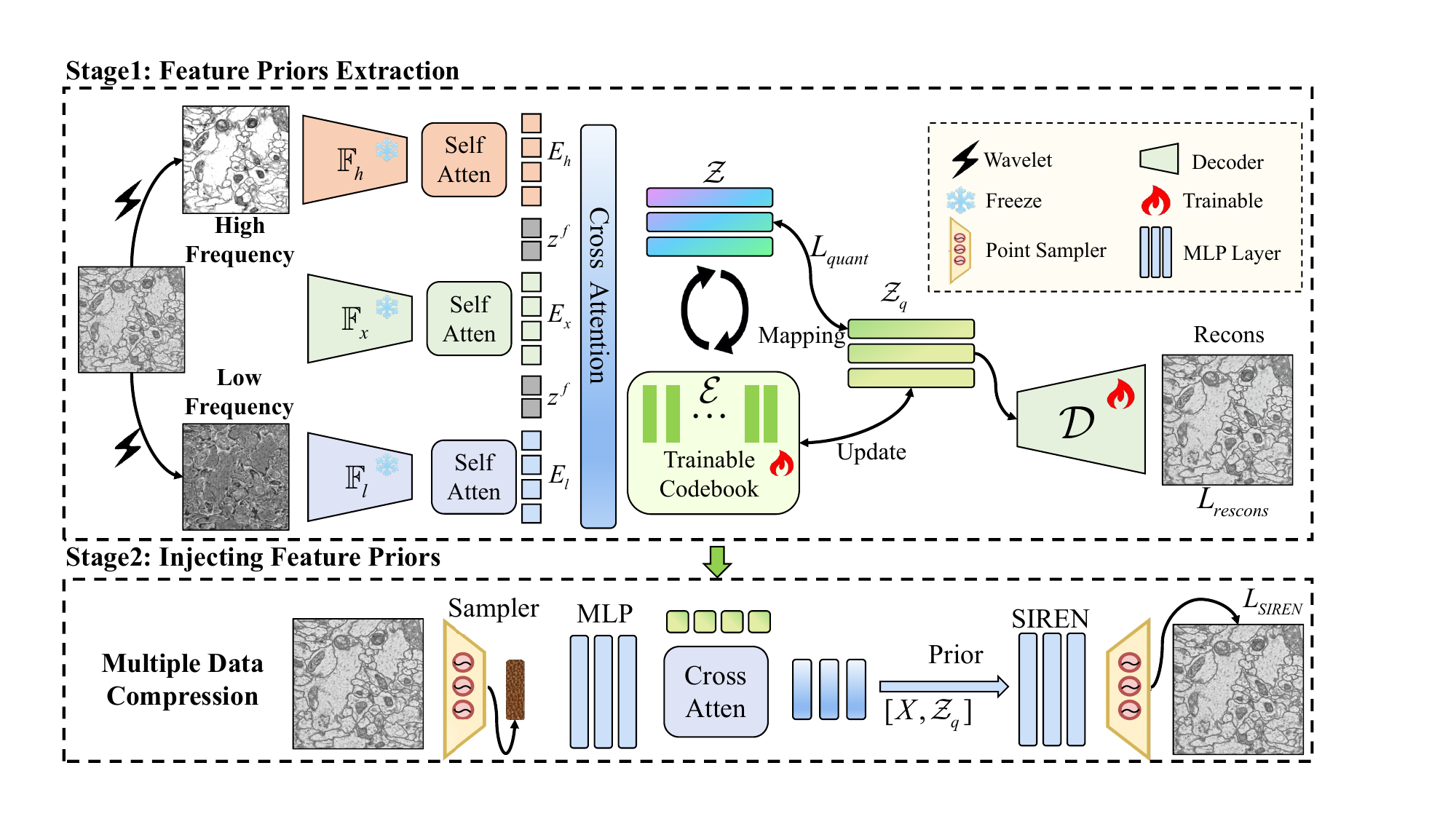}
    \vspace{-0.4cm}
    \caption{Schematic representation of a wavelet-based prior knowledge infusion into an INR compression network. The process involves the extraction of high, global, and low-frequency components using wavelet transform. These components are individually refined through self-attention mechanisms and collectively integrated via cross-attention with the INR network. The architecture further processes the information through transformer blocks and an MLP for effective compression of neural representations.}
    \label{fig:main_teacher}
\end{figure*}
INR-based methods like SCI \cite{yang2023sci} and TINC \cite{yang2023tinc} offer flexible block-wise approaches. COMBINER \cite{guo2023compression} compresses data as INRs using relative entropy encoding. These methods are more flexible and universal across different modalities.

Despite their flexibility, INRs face challenges like low processing efficiency and high-frequency information loss. To address these, we integrate medical image characteristics into INR inputs through robust visual feature extraction and process multiple data blocks simultaneously. This approach is largely unexplored in INR compression for large-scale medical imaging.

\section{Methodology}

\subsection{Problem Description}
We tackle the problem of representing multiple volumetric medical images using a conditional INR network. Our objective is to encode a set of N volumetric medical images $\{\mathcal{V}_1, \mathcal{V}_2, ..., \mathcal{V}_N\}$, where each $\mathcal{V}_i$ represents a 3D image block, into a compressed format.

We begin by conducting a theoretical analysis to compare the bounds on reconstruction errors between INR and Variational Autoencoder (VAE) methods. Our analysis reveals that if the INR decoder generates images that are closer to the original images with smaller magnitudes, then the INR will have a lower reconstruction error (See in Section \ref{sec:mae}). Based on this insight, we propose introducing prior frequency information into the INR framework.

The process begins with a specialized network, denoted as $\mathbb{F}$, which is designed to extract prior features $\mathcal{P}_i$ from each image block $\mathcal{V}_i$. These prior features $\mathcal{Z}_i$ encapsulate essential structural and intensity characteristics of the volumetric data.

Given a set of spatial coordinates $\mathbf{x} \in \mathbb{R}^3$, the INR network $\mathcal{F}$ is trained to predict the intensity value at $\mathbf{x}$, utilizing both the coordinates and the extracted prior features. The mapping function of the INR network is described as
\begin{equation}
    \mathcal{I}(\mathbf{x}) = \mathcal{F}({x_i}, \mathcal{Z}_i; \Theta),
    \label{euq_total}
\end{equation}
where $\mathcal{I}(\mathbf{x})$ represents the predicted intensity value at coordinate $\mathbf{x}$, $\mathcal{Z}_i$ are the prior features corresponding to $\mathcal{V}_i$, and $\Theta$ denotes the parameters of the INR network.

The compression ratio is determined by considering both the dimensions of the network parameters $\Theta$ and the dimensions of the prior features $\mathcal{Z}$. The aim is to optimize the INR network such that the reconstruction error between the original images $\mathcal{V}_i$ and their reconstructions from the compressed representation is minimized, while also maximizing the compression ratio. The framework is shown in Figure \ref{fig:main_teacher}.

\subsection{Feature Quantization and Extraction}
Our method combines wavelet transformation, pre-trained feature extraction, and codebook-based quantization to efficiently represent and reconstruct 3D medical images. The key steps are as follows:

\textbf{Wavelet Transformation.} We apply wavelet transformation to 3D medical images, extracting high-frequency, low-frequency, and original image features separately. The transformation is represented as
\begin{equation}
% \resizebox{\linewidth}{!}{$
\begin{aligned}
W_{high}(x, y, z) &= \int \hspace{-0.1cm}\int \hspace{-0.1cm}\int I(x', y', z') \cdot \psi_{h}(x-x', y-y', z-z') \,dx' \,dy' \,dz' ,\\
W_{low}(x, y, z) &= \int\hspace{-0.1cm}\int\hspace{-0.1cm}\int I(x', y', z') \cdot \psi_{l}(x-x', y-y', z-z') \,dx' \,dy' \,dz' ,
\end{aligned}
% $}
\label{waveequ}
\end{equation}
where $W_{high}(x, y, z)$ represents high-frequency information, $W_{low}(x, y, z)$ represents low-frequency information, $\psi_{h}$ and $\psi_{l}$ are high-pass and low-pass wavelet functions.

To obtain high-dimensional 3D features, we employ a pre-trained network $\mathbb{F}$ using a self-supervised method, as proposed by \cite{chen2023self}, which aligns multi-modal information through a contrastive learning approach. We keep the weights of $\mathbb{F}$ frozen during training. The extracted features are denoted as
\begin{equation}
w_{h}=\mathbb{F}(W_{h}), w_{l}=\mathbb{F}(W_{l}), w_x=\mathbb{F}(W_x),
\label{equ:extract}
\end{equation}
where $w_{h}$, $w_{l}$, and $w_x$ represent the high-frequency feature, low-frequency feature, and the raw image, respectively.

\textbf{Multimodal Feature Fusion.}
After obtaining high-dimensional features from images, we first utilize an MLP layer to transform these features into low-dimensional embeddings. We then employ a hybrid hierarchical attention mechanism for the fusion of multimodal features. To ensure feature distinctiveness in different modalities, we initially calculate the attention weights as $\mathbf{E}_{h}, \mathbf{E}_{l}, \mathbf{E}_x$ using a self-attention mechanism ($\mathbf{E} = MSA(LN(w_i))+w_i$). Subsequently, the resulting tokens are concatenated, and special bottleneck tokens $z^f$ are inserted between each modality for differentiation. The final weighted attention is computed as 
\begin{equation}
    \mathcal{Z} = MCA([\mathbf{E}_h, \mathbf{E}_l, \mathbf{E}_x] + pos \mid z^f),
\end{equation}
where $\mathbf{E}_h$, $\mathbf{E}_l$, and $\mathbf{E}_x$ represent the tokens from high-frequency, low-frequency, and raw image features, respectively, and $pos$ and $z^f$ are learnable position embeddings and bottleneck tokens. $MCA$ is the cross-attention mechanism. $LN$ means layer normalization.

\textbf{Codebook-based Quantization.} 
 The process of quantizing continuous features into a discrete codebook serves dual purposes: firstly, it reduces storage requirements, and secondly, it results in more consistent and stable feature representations. We employ a learnable codebook of dimension \(d\), mapping the prior features \(\mathcal{P}\) to the nearest vector (codeword) in the codebook. The quantization layer can be represented as
\begin{equation}
    z_q = \arg \min_{e_i \in \mathcal{E}} \| z - e_i \|,
    \label{equ_map}
\end{equation}
where \(\mathcal{E}\) is the set of all codewords in the codebook, and \(z_q\) is the quantized representation. The decoder $\mathcal{D}$ receives this quantized latent representation \(z_q\).

The final loss function is bifurcated into two segments: reconstruction loss and quantization loss. The reconstruction loss is employed to minimize the disparity between the input data $X$ and its reconstruction, expressed as 
\begin{equation}
    L_{MSE}=\| \mathcal{D}(z_q)-X_i \|^2.
\end{equation}
The quantization loss in our model is inspired by VQ-GANs \cite{esser2021taming}, but with a significant modification to suit our architecture. As our feature extractor is a fixed pre-trained network, we only utilize the first component of the VQ-GAN quantization loss, which is designed to stabilize the quantization process. This component, encoder output and quantized vector discrepancy, is given by
\begin{equation}
    L_{quant} = \| \text{sg}[z] - z_q \|^2,
\end{equation}
where \( \text{sg} \) denotes the stop-gradient operation.
Consequently, the overall loss function for our model is modified to
\begin{equation}
    L_{t} = L_{recon} + \gamma L_{quant},
\end{equation}
where \( \gamma \) is a weighting factor that balances the reconstruction and the first component of the quantization loss.

\subsection{Compression Knowledge Distillation}
\begin{wrapfigure}{r}{0.5\textwidth}
    \centering
    \includegraphics[width=\linewidth]{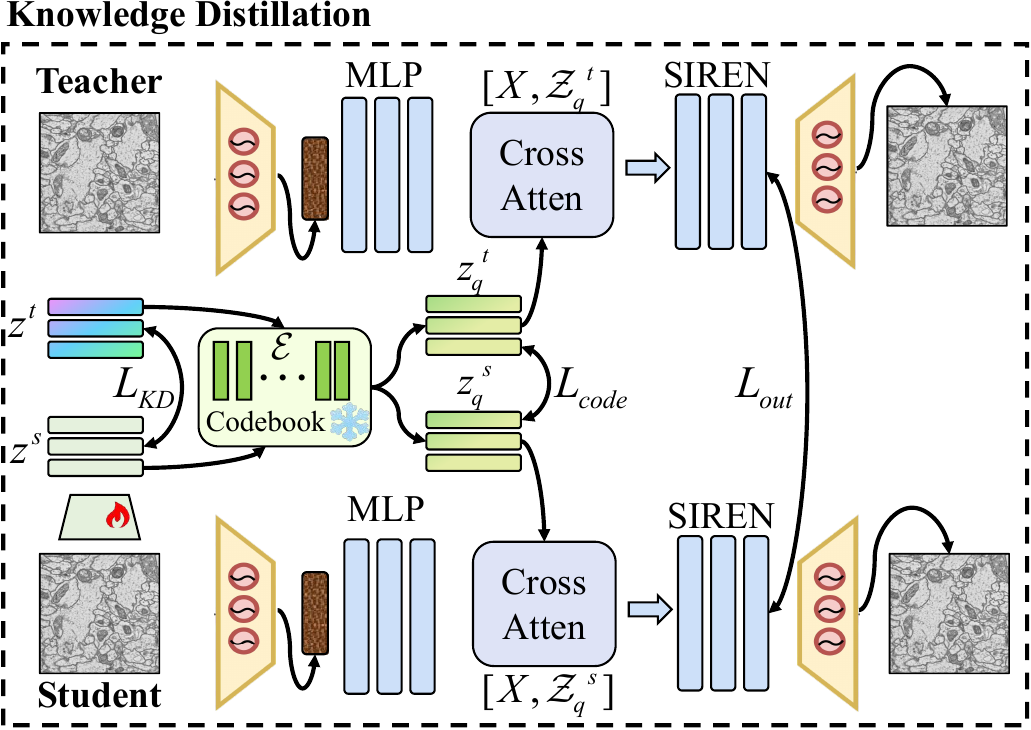}
    \vspace{-0.4cm}
    \caption{Knowledge distillation pipeline: student model mimics teacher using cross-attention for feature transfer, followed by parallel transformer processing, enhancing speed and compression while maintaining accuracy.}
    \vspace{-0.2cm}
    \label{fig:student}
\end{wrapfigure}
As delineated in Equation \ref{euq_total}, our compression strategy fundamentally relies on the discrete prior features $\mathcal{Z}_q$ and continuous coordinate values $X$, employing an INR network to fit a smooth function. In this context, we utilize a SIREN network for fitting purposes. Furthermore, we employ a knowledge distillation approach to condense the prior information from a more complex teacher model into a simplified student model. Our distillation process is bifurcated into two segments: aligning the intermediate discrete feature outputs and the final intensity estimation outputs between the teacher-student models as shown in Figure \ref{fig:student}.

\textbf{Intermediate Feature Alignment.}
In the first segment, the teacher model maintains its architecture \(\mathcal{F}^t\) as previously established, while the student model \(\mathcal{F}^s\) is subjected to a simplification process. We utilize the identical pre-trained network \(\mathbb{F}\) as a feature extractor for both models, the teacher model distinctively extracts high-frequency, low-frequency, and original image information. In contrast, the student model focuses solely on extracting features from the original image. Diverging from the teacher model's approach, we have designed the feature extractor of the student model to be trainable. This ensures that the feature remains consistent with that of the teacher model. Both teacher and student models share a common codebook \(\mathcal{E}\) for nearest neighbor mapping as equation \ref{equ_map}. Consequently, the loss is bifurcated into two pivotal components: discrete and continuous alignment. The discrete alignment is articulated as
\begin{equation}
    L_{code} = \| z^s - \text{sg}[z_q^s] \|^2 + \| z_q^t - z_q^s \|^2,
\end{equation}
where \( z_q^t \) represents the discrete codebook of the teacher model, and \( z_q^s \) corresponds to that of the student model. On the other hand, continuous alignment is described by the equation
\begin{equation}
L_{KD} = \tau^2 \hspace{-0.1cm}\cdot\hspace{-0.1cm} KL\left( \text{softmax}\left(\frac{z^t}{\tau}\right), \text{softmax}\left(\frac{z^s}{\tau}\right) \right),
\end{equation}
where  \( z^t \) and \( z^s \) represent the logits outputs of the teacher and student models, respectively. The temperature parameter \( \tau \) plays a pivotal role in moderating the degree of softness in the probability distributions. The Kullback-Leibler (KL) divergence, denoted as \( KL \), is utilized to measure the divergence between these probability distributions.

In our approach, the discrete alignment employs a Mean Squared Error (MSE) constraint, characterized as a hard constraint. In contrast, for continuous alignment, we apply a soft constraint. The primary objective of this soft constraint is to enable the student model to learn the probability distribution of the teacher model. This methodology of applying a soft constraint is particularly advantageous for model compression, allowing for a more efficient distillation process while preserving the essential characteristics of the teacher model’s predictive distribution. 
% This balance of hard and soft constraints facilitates an effective and nuanced transfer of knowledge from the teacher to the student model, optimizing the student's performance in a compressed framework.

\textbf{Final Output Alignment.}
In the second segment, the alignment of the final outputs is meticulously orchestrated by matching the outputs of both the teacher and student models, as well as aligning them with the true values. This crucial step ensures a harmonized and coherent relationship among the student model, the teacher model, and the actual data. The alignment loss is mathematically formulated as 
\begin{equation}
% \resizebox{\linewidth}{!}{$
\begin{aligned}
    L_{out}=&\|\mathcal{F}^s(x_i, \mathcal{Z}_i^s; \Theta^s) - \mathcal{F}^t(x_i, \mathcal{Z}_i^t; \Theta^t) \|^2 \\
    +&\|\mathcal{F}^s(x_i, \mathcal{Z}_i^s; \Theta^s) - X_i \|^2 + \|\mathcal{F}^t(x_i, \mathcal{Z}_i^t; \Theta^t) - X_i \|^2,
\end{aligned}
% $}
\end{equation}
where \(X_i\) denotes the intensity of the original pixel in the input. Building upon this framework, the total loss function of our proposed methodology is encapsulated as
\begin{equation}
    L_{total}=L_{code} + \beta_1 L_{KD} + \beta_2 L_{out},
\end{equation}
where \(\beta\) represents the weighting coefficients assigned to each constituent loss function. Our training regimen commences with the initial training of the codebook and the teacher model. Subsequently, we harness the knowledge distillation technique to refine the training of the student model and further fine-tune the teacher model.

\section{Experiments and Results}

\subsection{Implementation Details}
\label{details:sec}
\textbf{Experiment Settings}
Our teacher model was trained on a cluster equipped with 8 NVIDIA RTX 3090 GPUs, with each GPU handling a batch size of 4. We employed three feature extractors, all pre-trained ResNet-50 models, which were kept frozen during the training to extract features of dimension 512. The wavelet transform was implemented using the OpenCV library, and the codebook size was set to 256. The teacher model underwent 800 epochs of training using the AdamW optimizer with an initial learning rate of 1e-5, optimized with a cosine annealing learning rate strategy.

For the student model, we focused on extracting knowledge from the teacher model, using an active ResNet-50 as the feature extractor for the original images. The total network parameters were kept at half the size of the teacher model. The student model underwent fine-tuning for 400 epochs on data blocks learned by the teacher model. The training process was managed using the AdamW optimizer, with an initial learning rate of 1e-5. Detailed information regarding hyperparameter settings and specific regularization techniques will be provided in the supplementary materials and our code repository. During decoding, we needed to additionally save codebook information and indices, hence the computation of the compression ratio includes these two parts.

\textbf{Dataset.}
We evaluated the effectiveness of our method using CT and EM datasets. The CT data was sourced from the Medical Segmentation Decathlon (MSD) \cite{antonelli2022medical}, comprising CT scans of three organs: Liver, Colon, and Spleen. Each voxel typically measures 1.00$\times$1.00$\times$2.5 mm, suitable for medium-level detail organ recognition. These organs are represented by 201, 190, and 61 high-precision medical slices respectively, and we cropped each slice into 64$\times$512$\times$512 blocks centered for network input.

On the other hand, for the biomedical dataset, we used the CREMI dataset \cite{heinrich2018synaptic}, comprising serial section Transmission Electron Microscopy (ssTEM) images of Drosophila neural circuits. Since electron microscopy is highly sensitive to noise, these images contain many artifacts and missing sections. The voxel resolution is 4$\times$4$\times$40 nm, capturing the complex details necessary for neural structure reconstruction. The dataset includes three samples, each containing a volume of 125$\times$1250$\times$1250 voxels. Similarly, we selected a cropping size of 64$\times$512$\times$512.

\textbf{Baseline Methods.}
\label{baseline}
To quantitatively assess our model's performance, we employed two widely used metrics: Structural Similarity Index Measure (SSIM) and Peak Signal-to-Noise Ratio (PSNR). SSIM evaluates structural similarity, while PSNR measures pixel-level similarity, inversely related to Mean Squared Error (MSE). Our approach was benchmarked against traditional standards and neural representation methods, comparing it with state-of-the-art techniques.

We implemented standard codecs using FFMpeg and HM, including JPEG \cite{wallace1991jpeg}, JPEG XL \cite{alakuijala2019jpeg}, H.264 \cite{wiegand2003overview}, and HEVC \cite{sullivan2012overview}. Additionally, we evaluated data-driven models like DVC \cite{lu2019dvc} and SSF \cite{agustsson2020scale} with default settings. In neural representations, we considered SIREN \cite{sitzmann2020implicit}, NeRF \cite{martin2021nerf}, NeRV \cite{chen2021nerv}, and TINC \cite{yang2023tinc}. A key aspect of our analysis was highlighting neural representation techniques for precise compression ratio tuning. We used distortion curves and tested the CT dataset at compression ratios of 64$\times$, 128$\times$, 256$\times$, 512$\times$, and 1024$\times$, and the EM dataset at ratios of 4$\times$, 8$\times$, 12$\times$, 16$\times$, and 20$\times$.

For specific parameters, FFMpeg was used for JPEG with -q:v 2 for high quality and -q:v 10 for lower quality. For JPEG XL, we used -q 90 for high quality and -q 50 for lower quality. H.264 was encoded with FFMpeg at -crf 23 for medium quality and -crf 28 for lower quality. HEVC was configured in HM with -q 32 for medium quality and -q 37 for lower quality. These settings ensured a fair comparison of compression ratio and image quality.

\subsection{Results}
% \usepackage{multirow}
Our experimental evaluation is divided into two distinct scenarios: within-domain and cross-domain knowledge distillation, to comprehensively evaluate the effectiveness of our proposed method.
\begin{figure}[t]
    \centering
    \includegraphics[width=0.8\textwidth]{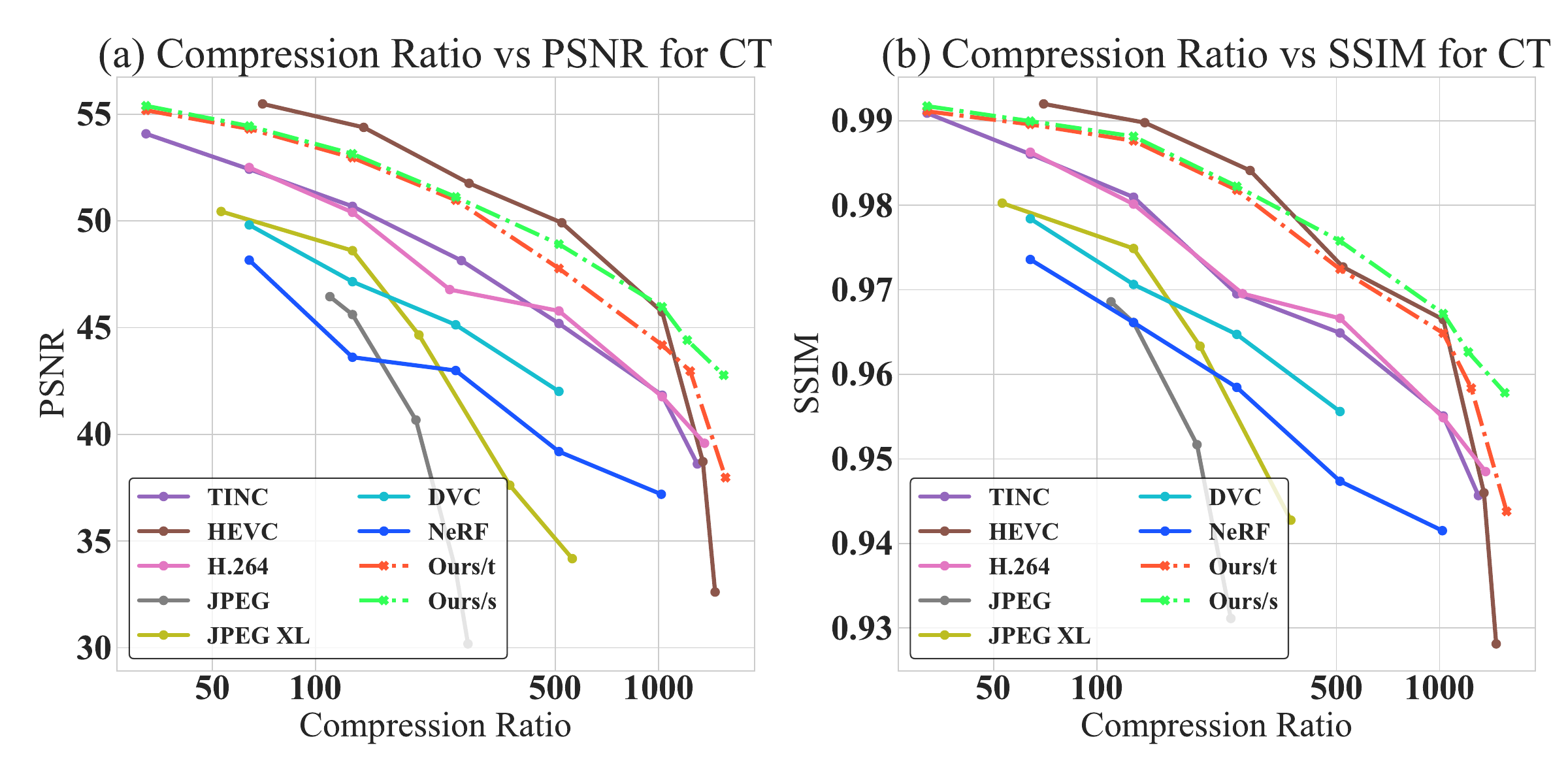}
    \vspace{-0.3cm}
    % \caption{cap1}
    \label{fig:ct_compress}
    \vspace{-0.4cm} % 调整两个图像之间的垂直间距
    \includegraphics[width=0.8\textwidth]{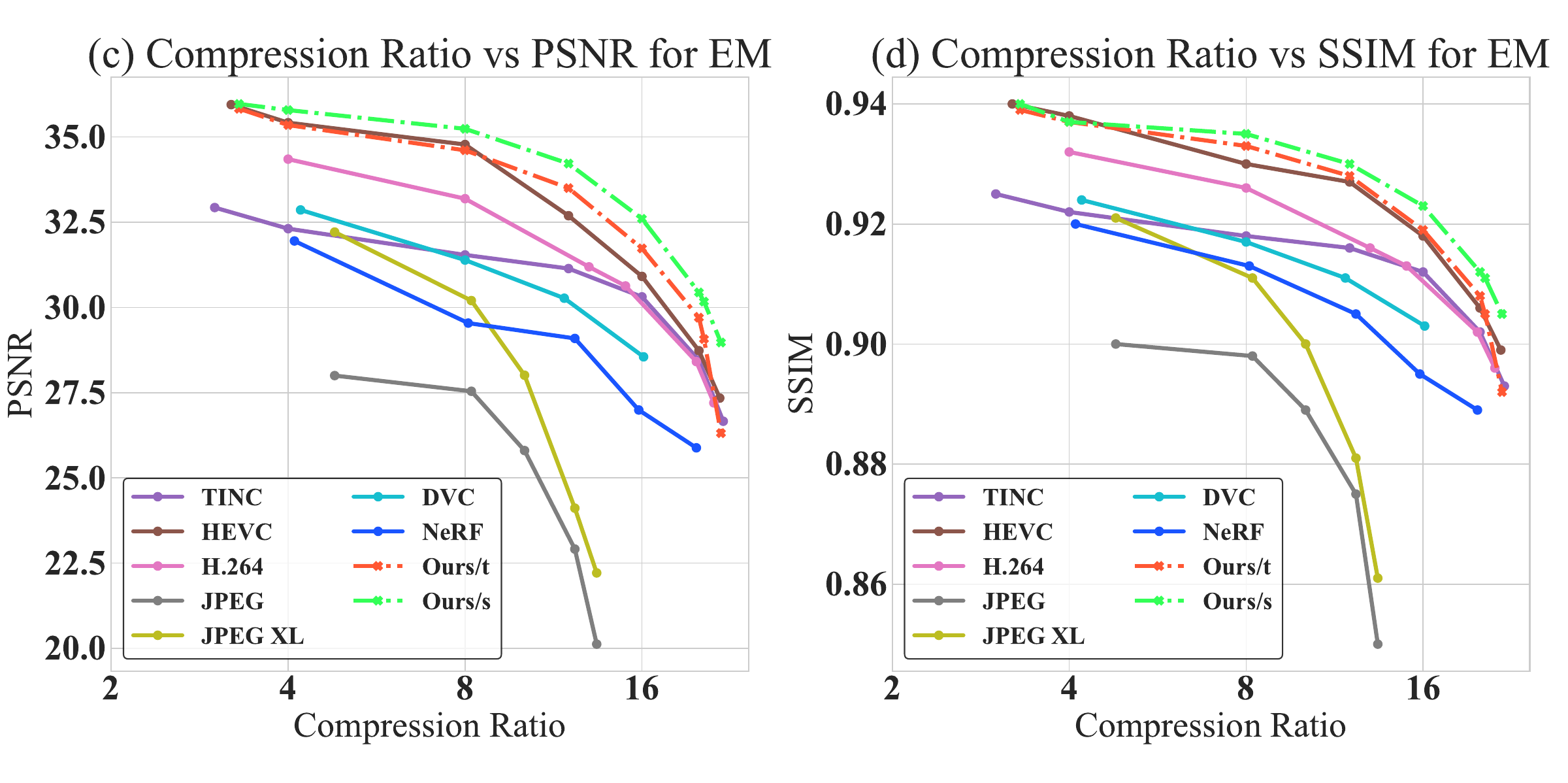}
    \caption{Distortion curves of Within-Domain Compression. The INR-based approach enables precise control over compression ratios, while other methods approximate using similar rates. `ours/t' denotes our teacher model, and `ours/s' represents our student model.}
    \label{fig:em_compress}
\end{figure}

\textbf{Within-Domain Results.} In our within-domain experiments, the teacher model is trained and distilled within selected CT or EM datasets. Initially, the teacher model is trained on a specific subset of the dataset, encompassing all inherent structural information of the dataset. This robust model is then utilized to distill knowledge into the student model, which is subsequently fine-tuned on similar types of datasets. The distillation process leverages the complex details absorbed by the teacher model, with a special emphasis on preserving key high-resolution features critical for reconstruction tasks. The performance of the student model is benchmarked against the teacher model, with a particular focus on the preservation of detailed structural information and compression efficiency. Our experimental results are presented in Table \ref{fig:main1}, and the performance curves across various compression ratios are shown in Figure \ref{fig:em_compress}.

In the CT dataset, our approach achieved the best performance among similar INR methods and is comparable to commercial methods like HEVC. The performance of the model significantly improved after knowledge distillation. Notably, our approach is more resistant to distortion at high compression ratios, whereas HEVC experiences a drastic performance drop beyond a 1000$\times$ compression ratio. In the CREMI dataset, our method surpassed HEVC, and we observed a significant decline in the performance of INR-based approaches like TINC, likely due to their network structures causing overly smooth image recovery. This is unreasonable for EM images rich in high-frequency information. Our method, by introducing frequency information as a prior, mitigates this over-smoothing issue inherent in INR. Our visualization results are displayed in Figure \ref{fig:CREMI}.
\begin{table*}[t]
\centering
\caption{Within-Domain Results. The label `w/o KD' denotes baseline teacher models without knowledge distillation, while `w/ KD' signifies student models post-knowledge distillation. The data in the table represent average results at varying compression ratios. Values in \textbf{bold} and \underline{underlined} formats indicate the top and second-best performances, respectively.
}
\fontsize{9}{9}\selectfont
\begin{tabular}{rcccccccc}
\toprule[1.2pt]
\multirow{2}{*}{Method} & \multicolumn{2}{c}{Liver} & \multicolumn{2}{c}{Colon} & \multicolumn{2}{c}{Spleen} & \multicolumn{2}{c}{CREMI} \\ \cline{2-9} 
                        & PSNR       & SSIM        & PSNR        & SSIM        & PSNR        & SSIM         & PSNR        & SSIM        \\ \midrule
TINC \cite{yang2023tinc}            & 45.81      & 0.9756      & 45.03       & 0.9744      & 45.35       & 0.9709       & 27.99       & 0.8427      \\
SIREN \cite{sitzmann2020implicit}          & 43.96      & 0.9602      & 44.54       & 0.9742      & 44.49       & 0.9689       & 28.54       & 0.8444      \\
NeRF \cite{martin2021nerf}       & 42.82      & 0.9528      & 39.37       & 0.9686      & 45.23       & 0.9740       & 27.70       & 0.8469      \\
NeRV \cite{chen2021nerv}           & 42.96      & 0.9546      & 39.91       & 0.9713      & 44.38       & 0.9700       & 26.71       & 0.8403      \\
JPEG \cite{wallace1991jpeg}           & 41.85      & 0.9654      & 37.36       & 0.9632      & 42.60       & 0.9610       & 25.26       & 0.8250      \\
JPEG XL \cite{alakuijala2019jpeg}                & 45.99      & 0.9773      & 44.12       & 0.9700      & 44.43       & 0.9645       & 28.71       & 0.8499      \\
H.264 \cite{wiegand2003overview}           & 44.03      & 0.9634      & 42.55       & 0.9724      & 45.41       & 0.9722       & 29.96       & 0.8478      \\
HEVC \cite{sullivan2012overview}            & \textbf{47.28}      & \textbf{0.9835}      & \textbf{48.21}       & \textbf{0.9851}      & \textbf{46.41}       & \textbf{0.9769}       & \underline{30.41}       & \underline{0.8547}      \\
DVC \cite{lu2019dvc}            & 42.48      & 0.9533      & 37.58       & 0.9664      & 43.34       & 0.9663       & 26.46       & 0.8360      \\
SSF \cite{agustsson2020scale}            & 45.05      & 0.9705      & 40.93       & 0.9714      & 45.42       & 0.9639       & 28.03       & 0.8449      \\ \midrule \midrule
UniCompress(w/o KD)     & 46.69      & 0.9791      & 46.90       & 0.9752      & 45.46       & 0.9730       & 30.25       & 0.8542      \\
UniCompress(w/ KD)     & \underline{47.08}      & \underline{0.9811}      & \underline{47.79}       & \underline{0.9812}      & \underline{46.32}       & \underline{0.9758}       & \textbf{31.02}       & \textbf{0.8601}      \\ \bottomrule[1.2pt]
\end{tabular}

\label{fig:main1}
\end{table*}

% \begin{figure*}[t]
%     \centering
%     \vspace{-0.2cm}
%     \includegraphics[width = \linewidth]{figures/compression_ratio_total2.pdf}
%     \vspace{-0.8cm}
%     \caption{Distortion curves of Within-Domain Compression. The INR-based approach enables precise control over compression ratios, while other methods approximate using similar rates. `ours/t' denotes our teacher model, and `ours/s' represents our student model.}
%     \vspace{-0.3cm}
%     \label{fig:1112}
% \end{figure*}

\begin{table}[t]
\centering
\caption{Cross-Domain Distillation Results. w/o KD: baseline teacher models. w/ KD: distilled student models. Target and source specify the respective data domains.}
\fontsize{9}{9}\selectfont
\renewcommand\tabcolsep{2.5pt}
    \begin{tabular}{rrcccc}
    \toprule[1.2pt]
    \multirow{2}{*}{Source} & \multirow{2}{*}{Target} & \multicolumn{2}{c}{PSNR} & \multicolumn{2}{c}{SSIM} \\ \cline{3-6} 
                            &                         & w/o KD      & w/ KD  & w/o KD      & w/ KD  \\ \midrule
    Colon                   & Liver                   & 46.69       & \cellcolor{lightgray}47.47 \gain{0.78}   & 0.9791      & \cellcolor{lightgray}0.9834 \gain{0.0043}  \\
    Spleen                  & Liver                   & 46.69       & \cellcolor{lightgray}47.12 \gain{0.43}   & 0.9791      & \cellcolor{lightgray}0.9821 \gain{0.0030}  \\
    Spleen                  & Colon                   & 46.90       & \cellcolor{lightgray}47.34 \gain{0.44}   & 0.9752      & \cellcolor{lightgray}0.9801 \gain{0.0049}  \\
    Liver                   & Colon                   & 46.90       & \cellcolor{lightgray}47.21 \gain{0.31}   & 0.9752      & \cellcolor{lightgray}0.9789 \gain{0.0037}  \\
    Liver                   & Spleen                  & 45.46       & \cellcolor{lightgray}45.89 \gain{0.43}   & 0.9730      & \cellcolor{lightgray}0.9741 \gain{0.0011}  \\
    Colon                   & Spleen                  & 45.46       & \cellcolor{lightgray}45.83 \gain{0.37}   & 0.9730      & \cellcolor{lightgray}0.9761 \gain{0.0031}  \\ \bottomrule[1.2pt]
    \end{tabular}
% \vspace{-0.3cm}
\label{tab:cross}
\end{table}

\textbf{Cross-Domain Results.}
In consideration of practical scenarios where the data requiring compression is often unknown, we ventured to extend the concept of knowledge distillation into the realm of cross-domain experiments. Specifically, we trained teacher models on datasets featuring different modalities and organs, subsequently distilling this prior knowledge into distinct organs or modalities to probe broader applicability and robustness. For our experiments, we selected CT datasets for mutual distillation, while EM datasets were excluded due to significant compression ratio differences. 

The experimental results, as shown in Table \ref{tab:cross}, illustrate the average outcomes under varying compression ratios. Intriguingly, our findings reveal that cross-domain distillation enables the student models to surpass the average performance metrics of the teacher models. Remarkably, we observed that for the Liver dataset, the gains from cross-domain distillation even outperformed those within the same domain. This was attributed to the Liver dataset containing more background elements compared to the other two datasets, which may lead to insufficient discriminative information being introduced during the initial training phase of the teacher model using only the Liver dataset, thereby rendering minimal gains from knowledge distillation. Our results imply that this methodology can rapidly adapt to new datasets, a capability crucial in practical applications.

\textbf{Comparison of Running Time.}
Traditional INR schemes, often tailored to overfit individual data points, significantly impeding encoding speed. Our method, UniCompress, achieves a 4$\sim$5 fold increase in encoding speed with only a marginal increase in decoding time. We have meticulously documented the time required for compression and decompression for each algorithm in our experiments, as presented in Supplementary Tables \ref{tab:time}. The enhanced encoding speed of UniCompress is particularly beneficial in handling large volumes of medical imagery, significantly reducing compression time, a critical advancement in medical applications.

\begin{figure}[t]
    \centering
    \includegraphics[width=\linewidth]{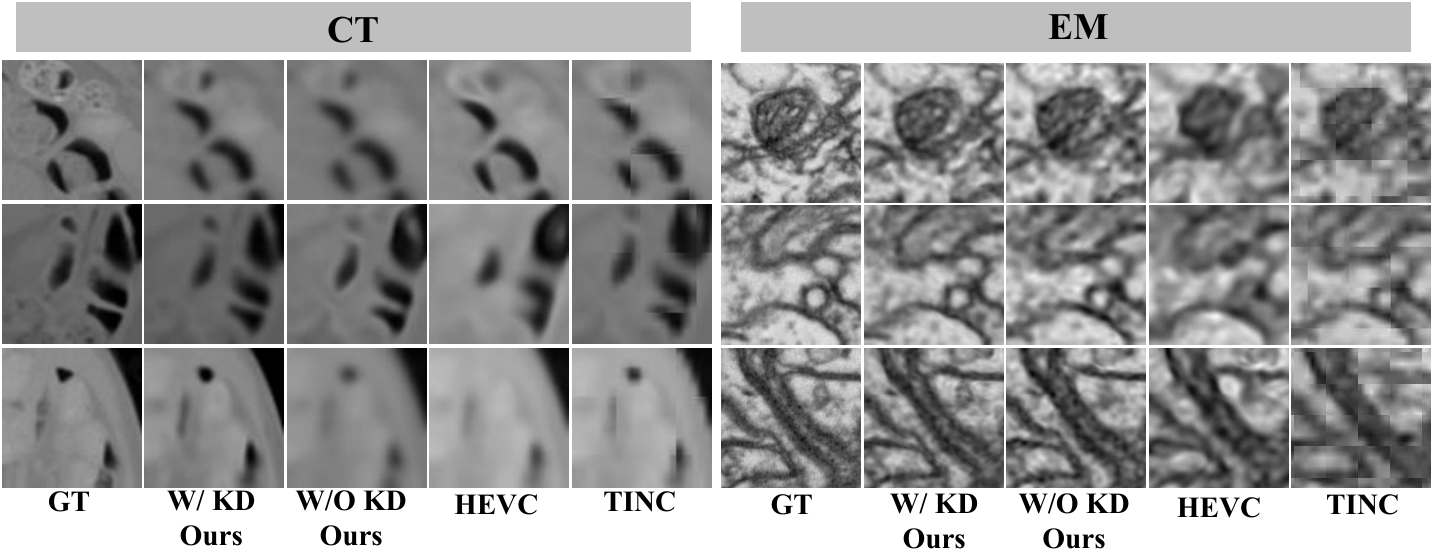}
    \vspace{-0.6cm}
    \caption{Visualization of CT and EM Image Compression. The figure displays the results of image compression for CT images at around 512$\times$ compression ratio and for EM images at around 12$\times$ compression ratio. }
    \vspace{-0.2cm}
    \label{fig:CREMI}
\end{figure}

\subsection{Ablation Studys}
\textbf{Frequency Domain Information Ablation Study.} We employed wavelet transform to treat frequency domain information as a codebook prior. To validate the efficacy of this information, we conducted an ablation study as shown in Table \ref{tab:fre}. This table displays the performance of the un-distilled teacher model. We observed that all frequency domain information contributes to certain improvements. Notably, for complex images like those in CREMI, the high-frequency information showed a significant enhancement, yielding a 0.94 dB gain in PSNR compared to the baseline method.

\textbf{Codebook Ablation Study.} In our prior knowledge, additional feature dimensions $v$ and a larger codebook size $d$ contribute to more refined discriminative features. However, due to the need to store the codebook and indices at the decoding end, an excessively large codebook size can lead to a reduction in compression ratio. For downstream compression tasks, there exists a certain trade-off in setting these parameters. Therefore, we conducted an ablation study on the size of the codebook, as shown in Table \ref{tab:codebook}. We tuned parameters that are suitable across two different modalities and found that a larger codebook size is more appropriate for complex EM data.

\textbf{Ablation Study on knowledge distillation.}
We employ discrete codebook features, continuous intermediate features, and the final output to constrain the knowledge distillation process. An ablation study was conducted on these three components, with the results presented in Table \ref{tab:knowledge_distillation}. We observed that constraints on the final output and the discrete codebook have a more significant impact on the results. This could be attributed to these components' ability to capture critical aspects of the learned representations, such as preserving crucial information through the distillation process and ensuring that the distilled model retains the essential characteristics of the teacher model. This emphasizes the importance of these constraints in achieving effective knowledge transfer.

\section{Conclusion}
In this study, we introduced ``UniCompress'', a novel approach for the concurrent compression of multiple medical image datasets, integrating wavelet transforms with INR networks. Our method effectively encodes both spatial and frequency domain information, significantly enhancing compression ratios and data integrity preservation. Additionally, our technique of knowledge distillation demonstrates considerable improvement in compression speed and cross-domain applicability. Our experiments on complex datasets like CREMI and MSD show that UniCompress achieves substantial improvements over existing methods in terms of compression performance and efficiency, laying the groundwork for future advances in medical image processing.

\clearpage
% \bibliography{example_paper}
\bibliographystyle{plain}

%%%%%%%%%%%%%%%%%%%%%%%%%%%%%%%%%%%%%%%%%%%%%%%%%%%%%%%%%%%%%%%%%%%%%%%%%%%%%%%
%%%%%%%%%%%%%%%%%%%%%%%%%%%%%%%%%%%%%%%%%%%%%%%%%%%%%%%%%%%%%%%%%%%%%%%%%%%%%%%
% APPENDIX
%%%%%%%%%%%%%%%%%%%%%%%%%%%%%%%%%%%%%%%%%%%%%%%%%%%%%%%%%%%%%%%%%%%%%%%%%%%%%%%
%%%%%%%%%%%%%%%%%%%%%%%%%%%%%%%%%%%%%%%%%%%%%%%%%%%%%%%%%%%%%%%%%%%%%%%%%%%%%%%
\clearpage
\appendix
\begin{center}
    \Large \textbf{Appendix}
\end{center}
% \section{Appendix / supplemental material}
\section{Theoretical Advantages of INR Compression}
\label{sec:mae}
\begin{assumption}\label{assump:image}
Suppose we have an image dataset $\mathcal{D}=\{x_i\}_{i=1}^n$ consisting of $n$ samples, where each sample $x_i \in \mathbb{R}^{H \times W \times C}$ is an image of size $H \times W$ with $C$ channels. We assume that these images are independently and identically distributed (i.i.d.) samples from some unknown distribution $p(x)$.
\end{assumption}

We consider using Implicit Neural Representations (INR) and Variational Autoencoders (VAE) for image compression. For INR, we define an encoder $f_{\theta}:\mathbb{R}^{H \times W \times C} \to \mathbb{R}^d$ that maps an image to a $d$-dimensional vector in the latent space, and a decoder $g_{\phi}:\mathbb{R}^d \to \mathbb{R}^{H \times W \times C}$ that maps a latent vector back to the image space. The training objective of INR is to minimize the reconstruction error:
\begin{equation}
\mathcal{L}_{\mathrm{INR}}(\theta, \phi) = \mathbb{E}_{x \sim p(x)}\left[\|x - g_\phi(f_\theta(x))\|_2^2\right].
\end{equation}

For VAE, we introduce a prior distribution $p(z)$ (usually taken to be a standard Gaussian), and define an encoder $q_{\psi}(z|x)$ and a decoder $p_{\omega}(x|z)$. The training objective of VAE is to maximize the evidence lower bound (ELBO):
\begin{equation}
\mathcal{L}_{\mathrm{VAE}}(\psi, \omega) = \mathbb{E}_{x \sim p(x)}\left[\mathbb{E}_{z \sim q_{\psi}(z \mid x)}\left[\log p_{\omega}(x \mid z)\right] - D_{\mathrm{KL}}\left(q_{\psi}(z \mid x) \| p(z)\right)\right].
\end{equation}

\begin{theorem}\label{thm:compression}
Under Assumption \ref{assump:image}, suppose INR and VAE have the same compression rate, i.e., $d_{\text{INR}} = d_{\text{VAE}} = d$. Denote the optimal parameters of INR and VAE as $(\theta^*, \phi^*)$ and $(\psi^*, \omega^*)$, respectively. Then,
\begin{equation}
\mathbb{E}_{x \sim p(x)}\left[\|x - g_{\phi^*}(f_{\theta^*}(x))\|_2^2\right] \leq \mathbb{E}_{x \sim p(x)}\left[\|x - \mathbb{E}_{z \sim q_{\psi^*}(z \mid x)}[p_{\omega^*}(x \mid z)]\|_2^2\right].
\end{equation}
In other words, under the same compression rate, the reconstruction error of INR is no greater than that of VAE.
\end{theorem}

\begin{proof}
For VAE, we have:
\begin{align*}
\mathbb{E}_{x \sim p(x)}\left[\|x - \mathbb{E}_{z \sim q_{\psi^*}(z \mid x)}[p_{\omega^*}(x \mid z)]\|_2^2\right] &= \mathbb{E}_{x \sim p(x)}\left[\|x\|_2^2\right] - 2\mathbb{E}_{x \sim p(x)}\left[\left\langle x, \mathbb{E}_{z \sim q_{\psi^*}(z \mid x)}[p_{\omega^*}(x \mid z)] \right\rangle \right] \\
&\quad + \mathbb{E}_{x \sim p(x)}\left[\|\mathbb{E}_{z \sim q_{\psi^*}(z \mid x)}[p_{\omega^*}(x \mid z)]\|_2^2\right].
\end{align*}
This follows from the law of total expectation and the fact that $\mathbb{E}[\|Y\|_2^2] \geq \|\mathbb{E}[Y]\|_2^2$ for any random vector $```markdown
Y$.

On the other hand, for INR, we have:
\begin{align*}
\mathbb{E}_{x \sim p(x)}\left[\|x - g_{\phi^*}(f_{\theta^*}(x))\|_2^2\right] &= \mathbb{E}_{x \sim p(x)}\left[\|x\|_2^2\right] - 2\mathbb{E}_{x \sim p(x)}\left[\left\langle x, g_{\phi^*}(f_{\theta^*}(x)) \right\rangle \right] \\
&\quad + \mathbb{E}_{x \sim p(x)}\left[\|g_{\phi^*}(f_{\theta^*}(x))\|_2^2\right].
\end{align*}

Now, let's compare the second terms in the above two expressions. For VAE, we have:
\begin{align*}
\mathbb{E}_{x \sim p(x), z \sim q_{\psi^*}(z \mid x)}\left[\left\langle x, p_{\omega^*}(x \mid z) \right\rangle \right] &= \mathbb{E}_{x \sim p(x), z \sim q_{\psi^*}(z \mid x)}\left[\left\langle x, \mathbb{E}_{x' \sim p_{\omega^*}(x' \mid z)}[x'] \right\rangle \right] \\
&\leq \mathbb{E}_{x \sim p(x), z \sim q_{\psi^*}(z \mid x)}\left[\|x\|_2 \cdot \|\mathbb{E}_{x' \sim p_{\omega^*}(x' \mid z)}[x']\|_2 \right] \\
&\leq \mathbb{E}_{x \sim p(x), z \sim q_{\psi^*}(z \mid x)}\left[\|x\|_2 \cdot \mathbb{E}_{x' \sim p_{\omega^*}(x' \mid z)}[\|x'\|_2] \right] \\
&= \mathbb{E}_{x \sim p(x)}\left[\|x\|_2 \cdot \mathbb{E}_{z \sim q_{\psi^*}(z \mid x)}\left[\mathbb{E}_{x' \sim p_{\omega^*}(x' \mid z)}[\|x'\|_2] \right] \right] \\
&\leq \mathbb{E}_{x \sim p(x)}\left[\|x\|_2 \cdot \sqrt{\mathbb{E}_{z \sim q_{\psi^*}(z \mid x), x' \sim p_{\omega^*}(x' \mid z)}[\|x'\|_2^2]} \right] \\
&= \mathbb{E}_{x \sim p(x)}\left[\|x\|_2 \cdot \sqrt{\mathbb{E}_{z \sim q_{\psi^*}(z \mid x)}[\mathbb{E}_{x' \sim p_{\omega^*}(x' \mid z)}[\|x'\|_2^2]]} \right],
\end{align*}
where the last inequality follows from Jensen's inequality.

For INR, we have:
\begin{align*}
\mathbb{E}_{x \sim p(x)}\left[\left\langle x, g_{\phi^*}(f_{\theta^*}(x)) \right\rangle \right] &= \mathbb{E}_{x \sim p(x)}\left[\|x\|_2 \cdot \|g_{\phi^*}(f_{\theta^*}(x))\|_2 \cdot \cos \angle(x, g_{\phi^*}(f_{\theta^*}(x))) \right] \\
&\geq \mathbb{E}_{x \sim p(x)}\left[\|x\|_2 \cdot \|g_{\phi^*}(f_{\theta^*}(x))\|_2 \cdot \cos \angle_{\max} \right],
\end{align*}
where $\angle(x, g_{\phi^*}(f_{\theta^*}(x)))$ denotes the angle between vectors $x$ and $g_{\phi^*}(f_{\theta^*}(x))$, and $\angle_{\max}$ is the maximum possible angle between any pair of vectors in the dataset.

Comparing the bounds for VAE and INR, we observe that if
\begin{equation}
\sqrt{\mathbb{E}_{z \sim q_{\psi^*}(z \mid x)}\left[\mathbb{E}_{x' \sim p_{\omega^*}(x' \mid z)}[\|x'\|_2^2]\right]} \leq \|g_{\phi^*}(f_{\theta^*}(x))\|_2 \cdot \cos \angle_{\max},
\end{equation}
then the second term in the INR bound will be larger than that in the VAE bound, implying a smaller reconstruction error for INR.

Now, let us compare the third terms in the VAE and INR bounds. For VAE, we have:
\begin{equation}
\mathbb{E}_{x \sim p(x)}\left[\mathbb{E}_{z \sim q_{\psi^*}(z \mid x)}[\|p_{\omega^*}(x \mid z)\|_2^2]\right] = \mathbb{E}_{x \sim p(x), z \sim q_{\psi^*}(z \mid x)}\left[\mathbb{E}_{x' \sim p_{\omega^*}(x' \mid z)}[\|x'\|_2^2]\right],
\end{equation}
while for INR, we have:
\begin{equation}
\mathbb{E}_{x \sim p(x)}\left[\|g_{\phi^*}(f_{\theta^*}(x))\|_2^2\right].
\end{equation}
If the expected squared norm of the reconstructed images under INR is smaller than that under VAE, i.e.,
\begin{equation}
\mathbb{E}_{x \sim p(x)}\left[\|g_{\phi^*}(f_{\theta^*}(x))\|_2^2\right] \leq \mathbb{E}_{x \sim p(x), z \sim q_{\psi^*}(z \mid x)}\left[\mathbb{E}_{x' \sim p_{\omega^*}(x' \mid z)}[\|x'\|_2^2]\right],
\end{equation}
then the third term in the INR bound will be smaller than that in the VAE bound.

In summary, under the conditions:
\begin{align*}
&\sqrt{\mathbb{E}_{z \sim q_{\psi^*}(z \mid x)}\left[\mathbb{E}_{x' \sim p_{\omega^*}(x' \mid z)}[\|x'\|_2^2]\right]} \leq \|g_{\phi^*}(f_{\theta^*}(x))\|_2 \cdot \cos \angle_{\max}, \\
&\mathbb{E}_{x \sim p(x)}\left[\|g_{\phi^*}(f_{\theta^*}(x))\|_2^2\right] \leq \mathbb{E}_{x \sim p(x), z \sim q_{\psi^*}(z \mid x)}\left[\mathbb{E}_{x' \sim p_{\omega^*}(x' \mid z)}[\|x'\|_2^2]\right],
\end{align*}
the reconstruction error of INR will be smaller than that of VAE under the same compression rate. These conditions essentially require that the INR decoder generates images with smaller expected squared norm and higher cosine similarity to the original images compared to the VAE decoder.

\end{proof}

The key idea behind this proof is to compare the bounds on the reconstruction errors of INR and VAE by analyzing the terms in their respective expressions. The conditions derived in the proof provide insights into when INR can achieve better compression performance than VAE. Intuitively, if the INR decoder can generate images that are more similar to the original images and have smaller magnitudes compared to the VAE decoder, then INR will have a lower reconstruction error.

It is worth noting that the conditions in the proof are sufficient but not necessary for INR to outperform VAE. In practice, the actual performance of INR and VAE will depend on various factors such as the network architectures, training procedures, and the specific dataset. Nevertheless, this theoretical analysis sheds light on the potential advantages of INR over VAE in terms of image compression.

\section{Dataset Information}
Our experiments encompass three CT datasets and one Electron Microscopy (EM) dataset. Table \ref{tab:data} provides detailed information on all these datasets. To standardize the input size for our network, we uniformly cropped the datasets to volumes of 64$\times$512$\times$512. The CREMI dataset can be accessed at \url{https://cremi.org/data/}, and the MSD dataset is available at \url{http://medicaldecathlon.com/}.

% You can have as much text here as you want. The main body must be at most $8$ pages long.
% For the final version, one more page can be added.
% If you want, you can use an appendix like this one.  

% The $\mathtt{\backslash onecolumn}$ command above can be kept in place if you prefer a one-column appendix, or can be removed if you prefer a two-column appendix.  Apart from this possible change, the style (font size, spacing, margins, page numbering, etc.) should be kept the same as the main body.
\begin{table}[h]
\centering
\caption{Detailed information about the datasets used in the main experiments, including resolution, quantity, and other relevant details.}
\fontsize{8.5}{10.8}\selectfont
\renewcommand\tabcolsep{1.5pt}
\begin{tabular}{ccccccccc}
\toprule[1.2pt]
Name      & Category       & Species          & Position & Characteristic & Bit depth & w/o sign & Resolution        & Slice counts \\ \midrule
CREMI     & EM, biological & adult drosophila & brain    & anisotropy     & 8         & unsigned & 125×1250×1250     & /            \\
CT-Liver  & CT, medical    & human            & liver    & anisotropy     & 32        & signed   & $\sim$588×512×512 & 201          \\
CT-Lung   & CT, medical    & human            & lung     & anisotropy     & 32        & signed   & $\sim$241×512×512 & 96           \\
CT-Spleen & CT, medical    & human            & spleen   & anisotropy     & 32        & signed   & $\sim$94×512×512  & 61           \\ \bottomrule[1.2pt]
\end{tabular}
\label{tab:data}
\end{table}
%%%%%%%%%%%%%%%%%%%%%%%%%%%%%%%%%%%%%%%%%%%%%%%%%%%%%%%%%%%%%%%%%%%%%%%%%%%%%%%
%%%%%%%%%%%%%%%%%%%%%%%%%%%%%%%%%%%%%%%%%%%%%%%%%%%%%%%%%%%%%%%%%%%%%%%%%%%%%%%
\section{Numerical Results of the Ablation Study}
In this section, we provide the tables containing the numerical results of the ablation experiments.
\begin{table}[h]
\centering
\begin{tabular}{cc}

% 第一个表格
\begin{minipage}{0.45\textwidth}
\centering
\caption{Ablation study on frequency information.}
\fontsize{8.5}{10.8}\selectfont
\renewcommand\tabcolsep{2pt}
\begin{tabular}{lllcccc}
\toprule[1.2pt]
\multicolumn{1}{l}{\multirow{2}{*}{High}} & \multirow{2}{*}{Low} & \multirow{2}{*}{Raw} & \multicolumn{2}{c}{Spleen} & \multicolumn{2}{c}{CREMI} \\ \cline{4-7} 
\multicolumn{1}{l}{}                      &                      &                      & PSNR        & SSIM         & PSNR        & SSIM        \\ \midrule
                                           &                      & $\checkmark$        & 44.55       & 0.9671       & 28.65       & 0.8451      \\
                                           & $\checkmark$         & $\checkmark$        & 45.13       & 0.9701       & 29.07       & 0.8491      \\
$\checkmark$                               &                      & $\checkmark$        & 45.19       & 0.9712       & 29.51       & 0.8521      \\
\rowcolor[HTML]{EFEFEF} 
$\checkmark$                               & $\checkmark$         & $\checkmark$        & 45.46       & 0.9730       & 30.25       & 0.8542      \\ \bottomrule[1.2pt]
\end{tabular}
\label{tab:fre}
\end{minipage}

& % 分隔符

% 第二个表格
\begin{minipage}{0.45\textwidth}
\centering
\caption{Ablation study on codebook size. `/' indicates using continuous feature as discriminative image prior.}
\fontsize{8.5}{10.8}\selectfont
\renewcommand\tabcolsep{2pt}
\begin{tabular}{llcccc}
\toprule[1.2pt]
                    &                     & \multicolumn{2}{c}{Spleen} & \multicolumn{2}{c}{CREMI} \\ \cline{3-6} 
\multirow{-2}{*}{$d$} & \multirow{-2}{*}{$v$} & PSNR        & SSIM         & PSNR        & SSIM        \\ \midrule
/                    & 512                 & 42.17       & 0.9541       & 27.89       & 0.8192      \\
128                  & 512                 & 43.32       & 0.9523       & 28.14       & 0.8287      \\
512                  & 1024                & 44.97       & 0.9651       & 30.11       & 0.8532      \\
128                  & 1024                & 43.13       & 0.9617       & 28.87       & 0.8207      \\
\rowcolor[HTML]{EFEFEF} 
256                  & 512                 & 45.46       & 0.9730       & 30.25       & 0.8542      \\ \bottomrule[1.2pt]
\end{tabular}
\label{tab:codebook}
\end{minipage}

\end{tabular}
\end{table}

\begin{table}[h]
\centering
\caption{Ablation study on knowledge distillation.}
\fontsize{8.5}{10.8}\selectfont
\renewcommand\tabcolsep{2pt}
\begin{tabular}{lllcccc}
\toprule[1.2pt]
\multirow{2}{*}{$L_{KD}$} & \multirow{2}{*}{$L_{code}$} & \multirow{2}{*}{$L_{out}$} & \multicolumn{2}{c}{Spleen} & \multicolumn{2}{c}{CREMI} \\ \cline{4-7} 
                           &                             &                            & PSNR        & SSIM         & PSNR        & SSIM        \\ \midrule
                           &                             & $\checkmark$               & 46.01       & 0.9715       & 30.79       & 0.8569      \\
                           & $\checkmark$                & $\checkmark$               & 46.27       & 0.9741       & 30.93       & 0.8591      \\
$\checkmark$               &                             & $\checkmark$               & 46.11       & 0.9722       & 30.84       & 0.8574      \\
\rowcolor[HTML]{EFEFEF}
$\checkmark$               & $\checkmark$                & $\checkmark$               & 46.32       & 0.9758       & 31.02       & 0.8601      \\ \bottomrule[1.2pt]
\end{tabular}
\label{tab:knowledge_distillation}
\end{table}

\section{Compression Efficiency}
Beyond visual metrics, the time consumed for data encoding and decoding is a crucial measure of a compression algorithm's efficacy. We evaluated the encoding and decoding times for all comparison baselines discussed in the main text. Specifically, the decoding time refers to the duration required to decode 1,000 data blocks. The results are presented in Table \ref{tab:time}. Our approach demonstrates a 4$\sim$5 times faster encoding process for both CT and EM datasets compared to common INR compression schemes like TINC and SIREN. This significant speed-up in encoding can greatly enhance efficiency in large-scale medical image compression, while incurring only a minimal increase in decoding time.
% \usepackage{multirow}
\begin{table}[h]
\centering
\fontsize{8.5}{10.8}\selectfont
\renewcommand\tabcolsep{1.5pt}
\caption{Comparison of encoding and decoding time. We categorize the methods into INR-based, commercial, and data-driven approaches. The encoding and decoding times for CT data are based on 512$\times$ compression, while for EM data, they are based on 12$\times$ compression. To emphasize differences, the decoding times are multiplied by 1000.}
\begin{tabular}{rr|cccc}

\toprule[1.2pt]
\multicolumn{2}{c|}{}                                    & \multicolumn{2}{c}{CT}                                        & \multicolumn{2}{c}{EM}                                        \\ \cline{3-6} 
\multicolumn{2}{c|}{}                                    & Compression                   & Decompression                 & Compression                   & Decompression                 \\
\multicolumn{2}{c|}{\multirow{-3}{*}{Method}}            & (Seconds)                     & x1k(seconds)                  & (Seconds)                     & x1k(seconds)                  \\ \midrule
                              & UniCompress(w/ KD) & \cellcolor[HTML]{EFEFEF}571.3 & \cellcolor[HTML]{EFEFEF}313.7 & \cellcolor[HTML]{EFEFEF}581.8 & \cellcolor[HTML]{EFEFEF}321.9 \\
                              & TINC \cite{yang2023tinc}                     & 2271                          & 241.9                         & 2672                          & 279.1                         \\
                              & SIREN \cite{sitzmann2020implicit}                    & 2013                          & 200.8                         & 2319                          & 223.7                         \\
                              & NeRV \cite{chen2021nerv}                      & 2412                          & 197.3                         & 2672                          & 201.9                         \\
\multirow{-5}{*}{INR based}   & NeRF \cite{martin2021nerf}                     & 526.7                         & 102.4                         & 769.6                         & 124.4                         \\ \midrule
                              & H.264 \cite{wiegand2003overview}                    & 1.730                         & 381.3                         & 2.130                         & 448.6                         \\
\multirow{-2}{*}{Commercial}  & HEVC \cite{sullivan2012overview}                     & 1.070                         & 697.3                         & 2.130                         & 658.9                         \\ \midrule
                              & DVC \cite{lu2019dvc}                      & 62.39                         & 387.1                         & 68.91                         & 399.4                         \\
\multirow{-2}{*}{Data driven} & SSF \cite{agustsson2020scale}                      & 6.987                         & 397.0                         & 6.127                         & 119.7                         \\ \bottomrule
\end{tabular}
\label{tab:time}
\end{table}

\section{Pretraining Approach}
In pursuit of extracting 3D visual representations from medical images, we introduce a vision-language pretraining approach that leverages large language models for generating textual descriptions of 3D images. This framework encompasses several key components and learning objectives, and can be summarized as Figure \ref{fig:sum} .

\begin{figure}[t]
    \centering
    \includegraphics[width = 0.6\linewidth]{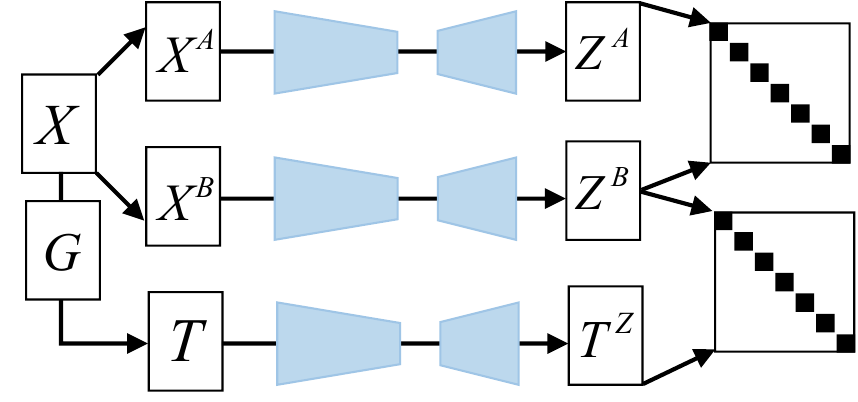}
    \caption{In this framework, \(X\) and \(T\) denote the medical images and their corresponding textual descriptions, respectively, while \(G\) represents the text generator. Our approach leverages large-scale models to generate textual descriptions, subsequently integrating image pretraining with vision-language pretraining methodologies.
}
    \label{fig:sum}
\end{figure}

\subsection{Text Generation}
The text generation process involves a generator \(G(\Theta)\), which is initially loaded with BLIP2's \cite{li2023blip} pretrained weights and further fine-tuned on the BIMCV dataset (a medical image-text pair dataset). This generator produces a sequence of tokens \(T = \{t_1, t_2, \ldots, t_n\}\) for each 3D medical image \(I\). The conditional probability of each token \(t_i\) is modeled as
\begin{equation}
P(t_i | I, t_{1:i-1}) = \text{softmax}(W_{o}h_i + b_o),
\end{equation}
where \(W_o\) and \(b_o\) are the weights and biases of the output layer, and \(h_i\) represents the hidden state at time step \(i\).

\subsection{Visual-Textual Representation Learning}
The framework employs an image encoder \(f_I(\cdot)\) and a text encoder \(f_T(\cdot)\) to learn visual-textual representations. The text encoder is initialized with the weights from BioBERT \cite{lee2020biobert} and is frozen during pretraining. For a batch of image-text pairs \((X_i, T_i)\), the feature representations are computed as \(v_{e,i} = f_I(X_i)\) and \(t_{e,i} = f_T(T_i)\). A contrastive learning objective is used to predict the matched pair \((v_{e,i}, t_{e,i})\) while ensuring significant separation of negative pairs.

\subsection{3D Visual Representation Learning}
To address the challenges of 1-to-n positive-negative pairings in 3D visual learning, we employ a negative-free learning objective. This involves generating two distinct views \(X_1\) and \(X_2\) of the medical volume \(X\) through data augmentation. The goal is to minimize the off-diagonal elements and maximize the diagonal elements of the cross-correlation matrix \(\hat{V}_{\text{corr}}\). The loss function is formulated as
\begin{equation}
L_{VLP} = \lambda_1 \left( \frac{1}{N} \sum_{i=1}^{N} \left( \frac{1}{2} \sum_{j \neq i} \hat{v}_{1,j}^T \hat{v}_{2,i} \right)^2 \right) + \lambda_2 \left( \frac{1}{N} \sum_{i=1}^{N} \left( \frac{1}{2} \sum_{j \neq i} \hat{t}_{1,j}^T \hat{t}_{2,i} \right)^2 \right),
\end{equation}
where \(\lambda_1\) and \(\lambda_2\) are non-negative hyperparameters, and \(\hat{v}_{1,j}\) and \(\hat{t}_{1,j}\) represent the normalized embeddings of the first and second views, respectively. Ultimately, the pretrained encoders are utilized to extract semantic information from high-frequency, low-frequency, and original images.

\subsection{Ablation Experiment of Pretraining}
In the main text, we mention the use of a frozen pretrained network for feature extraction. Freezing the network weights aids in stabilizing the training process of the model. Furthermore, vision-language pretraining methods have been validated as effective means for image feature extraction. In this section, we conduct ablation experiments, the results of which are shown in Table \ref{abl:pretrain}. We observe that activating the feature extraction weights during the training phase of the teacher model leads to failure in image reconstruction. This could be attributed to an insufficient number of epochs set, preventing adequate fitting to the codebook features. Meanwhile, the pretrained network substantially aids in feature extraction, significantly enhancing the reconstruction results.
% \usepackage{multirow}
\begin{table}[h]
\centering
\fontsize{8.5}{10.8}\selectfont
\renewcommand\tabcolsep{1.5pt}
\caption{Ablation experiment of pretraining.}
\label{abl:pretrain}
\begin{tabular}{cccccc}
\toprule[1.2pt]
\multirow{2}{*}{Freeze} & \multirow{2}{*}{Pretrain} & \multicolumn{2}{c}{Spleen} & \multicolumn{2}{c}{CREMI} \\ \cline{3-6} 
                        &                           & PSNR        & SSIM         & PSNR        & SSIM        \\ \hline
                        & $\checkmark$                      & 13.49       & 0.3142       & 12.71       & 0.2984      \\
$\checkmark$                    &                           & 41.51       & 0.9312       & 24.91       & 0.7841      \\
\rowcolor[HTML]{EFEFEF}
$\checkmark$                    & $\checkmark$                      & 45.46       & 0.9730       & 30.25       & 0.8542      \\ \bottomrule[1.2pt]
\end{tabular}
\end{table}
\section{Algorithm Pseudocode}
The source code for our paper is provided in the supplementary materials. For ease of understanding, this section presents the pseudocode format of the algorithm proposed in our paper as Algorithm \ref{al1} and \ref{al2}.

\begin{algorithm}[h]
\caption{Teacher Model Training with Mathematical Formulas}
\label{al1}
\SetKwInOut{Input}{Input}
\SetKwInOut{Output}{Output}

\Input{Medical images dataset}
\Output{Trained Teacher Model}

Initialize Teacher Model with pre-trained weights\;

\ForEach{image in dataset}{
    Apply wavelet transformation to image: $W(\text{image})$ \tcp*[h]{see Equ.\ref{waveequ}} \\
    Extract features using pre-trained network: $F = \text{ExtractFeatures}(W(\text{image}))$ \tcp*[h]{see Equ.\ref{equ:extract}} \\
    Quantize features using learnable codebook: $Q = \text{Quantize}(F)$ \\
    Compute reconstruction loss: $L_{\text{recon}} = \mathcal{L}_{\text{recon}}(F, Q)$ \\
    Compute quantization loss: $L_{\text{quant}} = \mathcal{L}_{\text{quant}}(F, Q)$ \\
    Total loss: $L = L_{\text{recon}} + \lambda L_{\text{quant}}$ \tcp*[h]{where $\lambda$ is a balance parameter} \\
    Update Teacher Model to minimize $L$
}
\end{algorithm}

\begin{algorithm}[h]
\caption{Knowledge Distillation with Mathematical Formulas}
\label{al2}
\SetKwInOut{Input}{Input}
\SetKwInOut{Output}{Output}

\Input{Teacher Model, Untrained Student Model}
\Output{Trained Student Model}

Initialize Student Model\;

\ForEach{training step}{
    Align Student's intermediate features with Teacher's: $F_{\text{student}} = \text{StudentFeatures}()$ \\
    Compute feature alignment loss: $L_{\text{feature}} = \mathcal{L}_{\text{feature}}(F_{\text{teacher}}, F_{\text{student}})$ \\
    Align Student's final output with Teacher's: $O_{\text{student}} = \text{StudentOutput}()$ \\
    Compute output alignment loss: $L_{\text{output}} = \mathcal{L}_{\text{output}}(O_{\text{teacher}}, O_{\text{student}})$ \\
    Total distillation loss: $L_{\text{distill}} = \alpha L_{\text{feature}} + \beta L_{\text{output}}$ \tcp*[h]{where $\alpha, \beta$ are balance parameters} \\
    Update Student Model to minimize $L_{\text{distill}}$
}
\end{algorithm}

% \begin{algorithm}[h]
% \caption{Optimization Algorithm}
% \label{alg:optimization}
% \SetKwInOut{Input}{Input}
% \SetKwInOut{Output}{Output}

% \Input{$x, \lambda, H, c, A, b, \text{lb}, \text{ub}$}
% \Output{$x, \lambda$}

% Initialize $x, \lambda$\;

% \While{not converged}{
%     Compute quadratic approximation: $Q(x) = 0.5 x^T H x + c^T x$\;
%     Linearize constraints: $A_l = A x + b$ \\
%     $A_u = A x - b$\;
%     Solve quadratic subproblem: $\min_x Q(x)$ \\
%     s.t. $A_l \leq x \leq A_u$ \\
%     $x \geq \text{lb}$ \\
%     $x \leq \text{ub}$\;
%     Compute step size: $\alpha = \min\left(1, \frac{\|x_k - x_{k-1}\|}{\|x_k\|}\right)$\;
%     Update $x$ and $\lambda$: $x_k = x_k - \alpha (x_k - x_{k-1})$ \\
%     $\lambda_k = \lambda_k - \alpha (\lambda_k - \lambda_{k-1})$\;
%     Check convergence: \\
%     \eIf{$\|\nabla Q(x_k)\| < \epsilon$}{
%         converged = True\;
%     }{
%         converged = False\;
%     }
% }
% \Return $x, \lambda$\;

% \end{algorithm}
% \begin{equation*}
% \max \sum_{i=1}^n \sum_{j=1}^m V_{ij} \cdot Y_{ij} - \sum_{i=1}^n \sum_{j=1}^m \sum_{k=1}^m C_{ijk} \cdot X_{ijk} + \sum_{i=1}^n \sum_{j=1}^m L_{ij} \cdot U_{ij}
% \end{equation*}
% $$
% \min _x \frac{1}{2} x^T H x+c^T x
% $$
% subject to:
% $$
% \begin{aligned}
% & A_l \leq x \leq A_u \\
% & x \geq l b \\
% & x \leq u b
% \end{aligned}
% $$
In this pseudocode, $\mathcal{L}_{recon}$, $\mathcal{L}_{quant}$, $\mathcal{L}_{feature}$, and $\mathcal{L}_{output}$ represent the specific loss functions used in the paper. Similarly, $\lambda$, $\alpha$, and $\beta$ are hyperparameters that need to adjust based on the specifics.
\section{Parameter Meaning and Hyperparameter Setting}
The meanings of all parameters mentioned in our paper are detailed in Table \ref{tab:params}. For the hyperparameters referenced in the main text, please refer to the code provided by us.

\section{Social Impact}
\label{sec:social}
\begin{flushleft}
In this study, we introduce UniCompress, a pioneering approach to multi-data medical image compression, combining Implicit Neural Representation (INR) networks with knowledge distillation techniques for enhanced efficiency and speed. We recognize that efficient image compression is vital for the sustainable development of the healthcare industry amid surging medical data volumes \cite{shen10harnessing}. UniCompress not only reduces storage and transmission costs but could also improve diagnostic efficiency by decreasing data transfer times, which is crucial for telemedicine and emergency medical scenarios.

Ethically, we prioritize patient privacy and data security in developing and applying UniCompress, ensuring the compression process does not compromise image quality, thus maintaining the accuracy and reliability of medical diagnoses.

\section{Limitations}
\label{limitations}
Despite demonstrating effectiveness across various medical image datasets, we acknowledge certain technical limitations of UniCompress. First, the current INR networks may not fully preserve details in images with high dynamic ranges and complex textures, especially in high-frequency areas. Second, quantization loss in the knowledge distillation process could lead to incomplete learning of the teacher model's knowledge in some cases. Moreover, our model's generalizability in handling large-scale datasets, especially diverse data from different medical institutions and devices, needs improvement.

\section{Future Works}
To address these limitations, we plan to explore: (1) developing new INR architectures for better capturing and retaining high-frequency image information while maintaining compression efficiency; (2) optimizing knowledge distillation algorithms for more precise knowledge transfer and reduced information loss; (3) researching adaptive quantization strategies to balance compression ratios and image quality; (4) enhancing model generalizability to adapt to varied medical image data sources and qualities; (5) exploring interpretability of deep learning models for medical professionals to understand the decision-making in the compression process.

We anticipate these future efforts will further enhance UniCompress's performance, making it a robust tool in medical image compression and contributing significantly to global health initiatives.
\end{flushleft}

\begin{table}[tb]
    \centering
    \fontsize{8.5}{10.8}\selectfont
\renewcommand\tabcolsep{1.5pt}
    \caption{List of symbols and their meanings.}
    \label{tab:params}
    \begin{tabular}{ll}
        \toprule[1.2pt]
        Symbol & Meaning \\
        \midrule
        $N$ & Number of volumetric medical images to compress \\
        $V_i$ & The $i$-th volumetric medical image \\
        $x \in R^3$ & Spatial coordinates \\
        $Z_i$ & Prior features of the $i$-th image block \\
        $\Theta$ & Parameters of the INR network \\
        $I(x)$ & Predicted intensity value at coordinate $x$ \\
        $W_{high}(x, y, z)$ & High-frequency information \\
        $W_{low}(x, y, z)$ & Low-frequency information \\
        $W_x$ & Raw image \\
        $F(\cdot)$ & Pre-trained network for feature extraction \\
        $wh$, $wl$, $wx$ & High-frequency, low-frequency, and raw image features, respectively \\
        $E_h$, $E_l$, $E_x$ & Attention weights for high-frequency, low-frequency, and raw image features, respectively \\
        $z_q$ & Quantized representation \\
        $\gamma$ & Weighting factor for quantization loss \\
        $L_{recon}$ & Reconstruction loss \\
        $L_{quant}$ & Quantization loss \\
        $L_{t}$ & Total loss function \\
        $L_{code}$ & Codebook alignment loss \\
        $L_{KD}$ & Knowledge distillation loss \\
        $L_{out}$ & Output alignment loss \\
        $L_{total}$ & Final total loss function \\
        $\tau$ & Temperature parameter for KL divergence \\
        $\beta$ & Weighting coefficients for balancing different loss functions \\
        $d$ & Dimension of the codebook \\
        $E$ & Set of all codewords in the codebook \\
        $z_s^q$, $z_t^q$ & Discrete codebook of the student and teacher models, respectively \\
        $z_s$, $z_t$ & Logits outputs of the student and teacher models, respectively \\
        $F_s$, $F_t$ & Student and teacher models, respectively \\
        $X_i$ & Input data \\
        \bottomrule[1.2pt]
    \end{tabular}
    \label{tab:symbols}

\end{table}
%%%%%%%%%%%%%%%%%%%%%%%%%%%%%%%%%%%%%%%%%%%%%%%%%%%%%%%%%%%%

%%%%%%%%%%%%%%%%%%%%%%%%%%%%%%%%%%%%%%%%%%%%%%%%%%%%%%%%%%%%

% \clearpage
% \input{Styles/checklist}

\end{document}